\documentclass{jams-l}

\usepackage{amssymb}

\usepackage{graphicx}

\newtheorem{theorem}{Theorem}[section]
\newtheorem{lemma}[theorem]{Lemma}

\theoremstyle{definition}
\newtheorem{definition}[theorem]{Definition}

\theoremstyle{remark}

\numberwithin{equation}{section}

\begin{document}

\title{The Ergodicity of the Collatz Process in Positive Integer Field}


\author{Bojin Zheng}
\address{College of Computer Science, South-Central University for Nationalities, Wuhan 430074, China}
\curraddr{School of Informatics and Computing, Indiana University, Bloomington, IN 47408, USA}
\email{zhengbojin@gmail.com}
\thanks{}

\author{Yangqian Su}
\address{College of Computer Science, South-Central University for Nationalities, Wuhan 430074, China}
\curraddr{College of Computer Science, South-Central University for Nationalities, Wuhan 430074, China}
\thanks{}

\author{Hongrun Wu}
\address{The State Key Laboratory of Software Engineering, Wuhan University, 430074, China}
\curraddr{The State Key Laboratory of Software Engineering, Wuhan University, 430074, China}
\email{ms.wuhr@gmail.com}
\thanks{}

\author{Li Kuang}
\address{The State Key Laboratory of Software Engineering, Wuhan University, 430074, China}
\curraddr{The State Key Laboratory of Software Engineering, Wuhan University, 430074, China}
\thanks{}



\date{}

\dedicatory{}

%
%
%

\begin{abstract}
The $3x+1$ problem, also called the Collatz conjecture, is a very interesting
unsolved mathematical problem related to computer science. This paper generalized this problem by
relaxing the constraints, i.e., generalizing this deterministic process to non-deterministic process, and set up three models. This paper analyzed the ergodicity of these models and proved that the ergodicity of the Collatz process in positive integer field holds, i.e., all the
positive integers can be transformed to 1 by the iterations of the Collatz
function.
\end{abstract}


\maketitle






\section{INTRODUCTION}

Since 1930s, researchers have deeply investigated the $3x+1$ problem. Until now, the $3x + 1$ problem has obtained many names, such as the Kakutani's problem, the Syracuse problem and Ulam's problem and so on\cite{1,3}.

The $3x+1$ problem can be stated from the viewpoint of computer algorithm as follows:

For any given integer $x$, if $x$ is odd, then let $x := 3x+1$; if $x$ is even, then
let $x:= x/2$; if we repeat this process, $x$ will certainly be $1$ at some time.

Mathematically, this problem can be presented as the iterations of a function $f(x)$, called the Collatz function shown as equation (\ref{eq.collatz}), i.e., $\forall x$, $\exists k$, $f^k(x)=1$, $x \in \mathbb{N}^{+}$. Here, $\mathbb{N}^{+}$ is the set of positive integers.

\begin{equation}
\label{eq.collatz}
f(x) = \left\{ {\begin{array}{l}
 3x + 1,  \mbox{ if }x\mbox{ is odd} \\
 \frac{x}{2}, \ \ \ \ \ \  \mbox{ if }x\mbox{ is even} \\
 \end{array}} \right.\\
 x > 0
\end{equation}

This problem is very hard to solve because the iteration process is very ``random'', although the Collatz function is deterministic.

In spite of the difficulty of this problem, the researchers still attained many fruitful achievements \cite{22}. From the view of probability theory \cite{23}, the researchers explored the existence of divergent trajectories; from the view of number theory and diophantine approximations and other mathematical tools, the researchers discussed the existence of the cycles other than $4 \rightarrow 2 \rightarrow 1$ \cite{24,6,13,15,19,21,9}; from the perspective of mathematical logic and theory of algorithms, the researchers studied the solvability of this problem \cite{10,7}. Moreover, this problem was also tried from the view of the fractal \cite{27,28}, graph theory \cite{20} and computation experiments \cite{14,1,22} and so on. Owing to the efforts of Prof. J. C. Lagarias, the related works were collected and commented \cite{a1}.

In this paper, we treated the Collatz function as a deterministic program (process), and generalized it to the non-deterministic program and set up three models; furthermore, we mapped the programs to the Collatz graphs. By the proposed models and the graph theory, we proved that the Collatz conjecture holds, i.e., all the positive integers can reach 1.

\section{MODELS}

The Collatz problem, which we called model $M0$, can be mapped to the Collatz graph \cite{3} as shown in Fig. \ref{fig.m0}.

\begin{figure}
\centerline{\includegraphics[width=6in,height=3.50in]{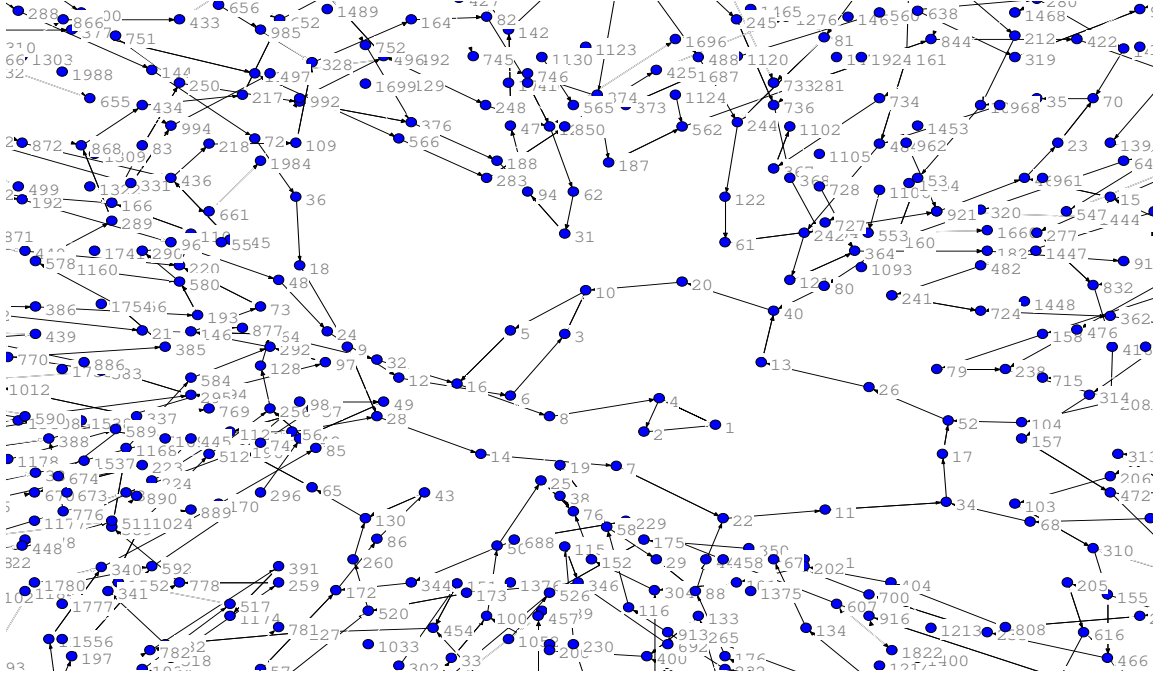}}
\caption{The original Collatz graph of model M0}
\label{fig.m0}
\end{figure}

In Fig. \ref{fig.m0}, every positive integer is a node. Every node has two directed edges.  Every edge responds to an item of Collatz function, which represents a transformation of the value of the variable $x$.

From the viewpoint of Collatz graph, the Collatz problem can be stated as follows: for any given node, i.e., positive integer $A$, there exists an $A \Rightarrow 1$ path in Collatz graph.

If there is an $A \Rightarrow 1$ path in Collatz graph, then we say $A$ is reachable. All the reachable positive integers will form a set $\Omega_0$, thus, the Collatz conjecture can also be stated as $\Omega_0 = \mathbb{N}^{+}$.

To solve this problem, we generalized the problem to non-deterministic process and set up three models. The first model is named $MS$, the second is named $M1$ and the third is named $M2$. Recall that the original problem is called $M0$.

In the model $MS$, we extended the Collatz function to three items, i.e., adding a non-deterministic item that is inverse to the $3x + 1$ item without the constraint of parity, shown in equation (\ref{eq.ms}).

\begin{equation}
\label{eq.ms}
f_s(x) = \left\{ {\begin{array}{l}
 3x + 1, \ \ \  \mbox{   if } x \bmod 2 \equiv 1  \\
 \frac{x}{2},\ \ \ \ \ \  \mbox{     if } x \bmod 2 \equiv 0  \\
 \frac{x - 1}{3}, \ \ \ \  \mbox{     if } x \bmod 3 \equiv 1  \\
 \end{array}} \right.\\
 x > 0
\end{equation}

By using a similar method to that of the Collatz graph, we can draw a graph to reflect equation (\ref{eq.ms}) as Fig. \ref{fig.ms}. In honor of L. Collatz, we called the generalized graphs as the Collatz graphs.

From Fig. \ref{fig.ms}, we can see that every positive integer, i.e., every node, has a few options to connect to the other nodes. For examples, node $7$ can connect to $22$ by $3x + 1$ or connect to $2$ by $\frac{x-1}{3}$; node $10$ can connect to $3$ by $\frac{x-1}{3}$ and $5$ by $\frac{x}{2}$. Therefore, the nodes in the Collatz graphs can be categorized into a few classes. Or say, the Collatz graphs have fruitful patterns.

\begin{figure}
\centerline{\includegraphics[width=6in,height=3.50in]{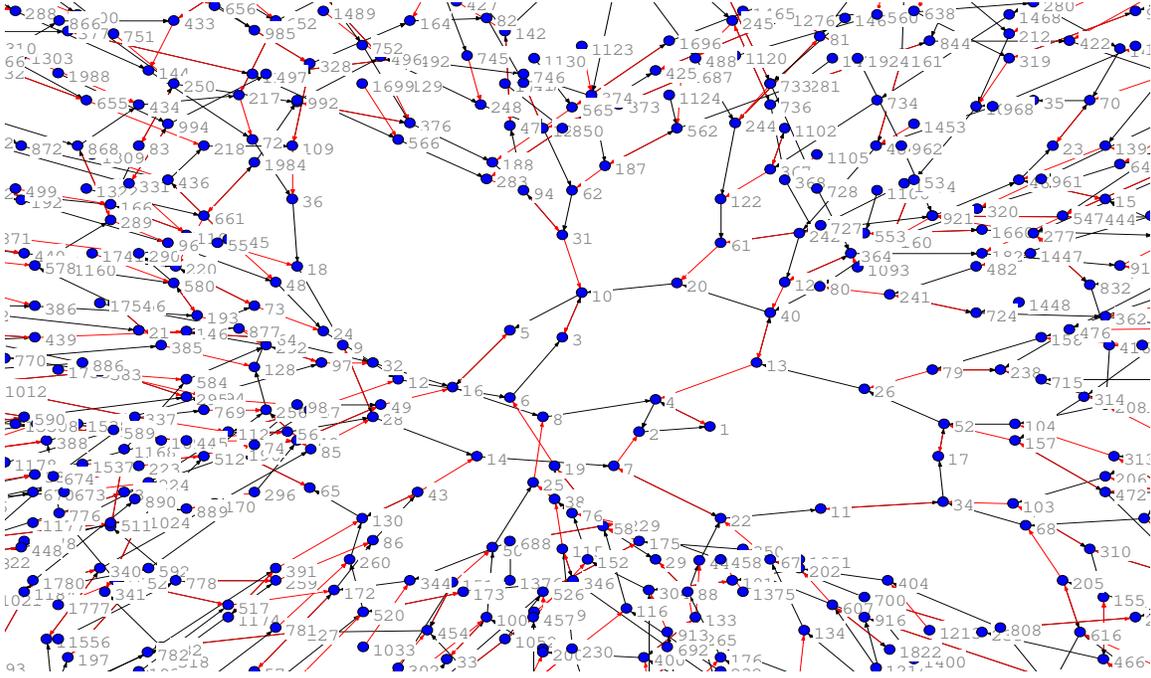}}
\caption{The Collatz graph of model MS}
\label{fig.ms}
\end{figure}

In Fig. \ref{fig.ms}, the red edges are different edges in contrast to Fig. \ref{fig.m0} of model $M0$.

Like $M0$, if there exists an $A \Rightarrow 1$ path, we say $A$ is reachable. We denote the reachable set of $MS$ as $\Omega_s$.

In the second model $M1$, we removed more constraints. The function of $M1$ is presented as equation (\ref{eq.m1}).

\begin{equation}
\label{eq.m1}
f_1 (x) = \left\{ {\begin{array}{l}
 3x + 1\mbox{ } \\
 \frac{x}{2}, \ \ \ \  \mbox{     if } x \bmod 2 \equiv 0 \\
 2x \\
 \frac{x - 1}{3},  \ \ \ \  \mbox{     if } x \bmod 3 \equiv 1 \\
 \end{array}} \right.\\
  x > 0
\end{equation}

Compared with the mode $MS$, model $M1$ adds a new item, i.e., $2x$, which is inverse to the item $x/2$.

With respect to equation (\ref{eq.m1}), the Collatz graph can be shown as in Fig. \ref{fig.m1}.

\begin{figure}
\centerline{\includegraphics[width=6in,height=3.50in]{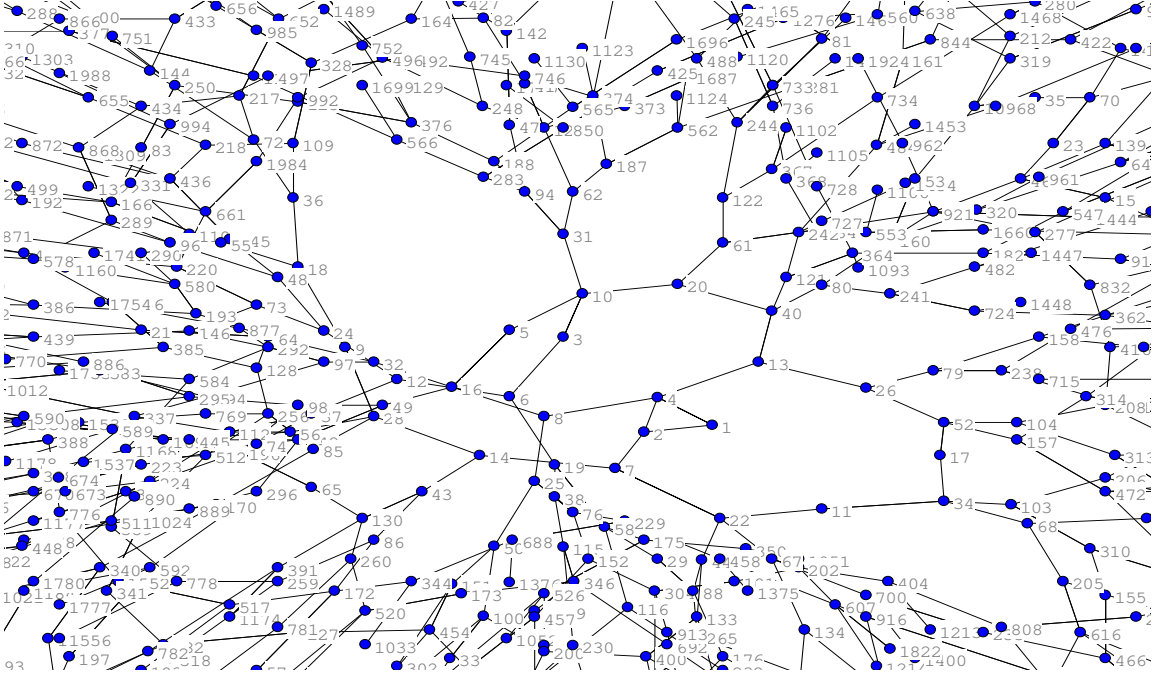}}
\caption{The Collatz graph of model M1}
\label{fig.m1}
\end{figure}

The structure of Fig. \ref{fig.m1} of model $M1$ is the same as Fig. \ref{fig.ms} of model $MS$ except that: some edges in Fig. \ref{fig.ms} are one-way, and all the edges in Fig. \ref{fig.m1} are two-way.

Similarly, we denote the reachable set of $M1$ as $\Omega_1$.

According to equation (\ref{eq.m1}), $f_1(x)$ has four items, which we call ``actions''.

\begin{definition}[Action] An action is an optional transformation of functions, i.e., an item of functions.
\end{definition}

We use $T$ for $3x+1$; $B$ for $x/2$; $F$ for $(x-1)/3$ and $D$ for $2x$, respectively.

According to the functions, action $T$ is inverse to $F$, and $B$ is inverse to
$D$.


We further generalize $M1$ to the third model $M2$, which can be formulated as equation (\ref{eq.m2}).

\begin{equation}
\label{eq.m2}
f_2 (x) = \left\{ {\begin{array}{l}
 3x + 1\mbox{ } \\
 \frac{x}{2}\mbox{ } \\
 2x \\
 \frac{x - 1}{3} \\
 \end{array}} \right.\\
  x > 0
\end{equation}

Obviously, in model $M2$, the functional values can be the rational
numbers, i.e., $M2$ has many nodes with the rational values.

Similarly, We denoted the reachable set of $M2$ as $\Omega_2$.

After the introduction of the proposed models, we further explored the properties of the proposed models. From model $M2$, we got some valuable clues for $M1$, and further we got the structural features of $M1$ and $MS$. Finally, we proved that $\Omega_0 = \mathbb{N}^{+}$.

\section{THE MODEL $M2$}

\begin{theorem}[The succession theorem]
The action sequence 'TDDFFBBT' is the succession function for any given positive integer, i.e., equation (\ref{eq.m2.x+1}) always holds for any given positive integer.
\begin{equation}
\label{eq.m2.x+1}
'TDDFFBBT'(x) = TBBFFDDT(x)= x + 1
\end{equation}
\end{theorem}


\begin{proof}
The calculations are listed as follows,

\[
\begin{array}{l}
 'TDDFFBBT'(x) = TBBFFDDT(x) \\
 = \frac{\frac{(x\ast 3 + 1)\ast 2\ast 2 - 1}{3} - 1}{3}\ast \frac{1}{2}\ast
\frac{1}{2}\ast 3 + 1 \\
 = \frac{4x}{3}\ast \frac{1}{4}\ast 3 + 1 \\
 = x + 1 \\
 \end{array}
\]
\end{proof}

\begin{lemma}
Every positive integer is reachable in M2, i.e., $\Omega_2 = \mathbb{N}^{+}$.
\end{lemma}

According to the succession theorem, $M2$ demonstrates a spiral structure. The positive integers are the central pillars.

\begin{theorem}[The 2 successions theorem]The action sequence 'DFFBTT' is the
succession of succession function for any positive integer, i.e., equation (\ref{eq.m2.x+2}) always holds for any given positive integer.
\begin{equation}
\label{eq.m2.x+2}
'DFFBTT'(x) = TTBFFD(x) = x + 2
\end{equation}
\end{theorem}

\begin{proof}
The calculations are listed as follows,

\[
\begin{array}{l}
 'DFFBTT'(x) = TTBFFD(x) \\
 = \left( {\frac{\frac{x\ast 2 - 1}{3} - 1}{3}\ast \frac{1}{2}\ast 3 + 1}
\right)\ast 3 + 1 \\
 = \left( {\frac{2x - 4}{9}\ast \frac{3}{2} + 1} \right)\ast 3 + 1 \\
 = x + 2 \\
 \end{array}
\]
\end{proof}

According to the 2 successions theorem, there often exist shorter action sequences to transform a number to another number in $M2$.

\begin{theorem}[The 3 successions theorem] The action sequence 'DDFFBBTT' is the
succession of succession of succession function for any given positive integer,
i.e., equation (\ref{eq.m2.x+3}) always holds for any given positive integer.
\begin{equation}
\label{eq.m2.x+3}
'DDFFBBTT'(x)= TTBBFFDD(x) = x + 3
\end{equation}
\end{theorem}

\begin{proof}
The calculations are listed as follows,

\[
\begin{array}{l}
'DDFFBBTT'(x)= TTBBFFDD(x) \\
 = \left( {\frac{\frac{x\ast 2\ast 2 - 1}{3} - 1}{3}\ast \frac{1}{2}\ast
\frac{1}{2}\ast 3 + 1} \right)\ast 3 + 1 \\
 = \left( {\frac{4x - 4}{9}\ast \frac{3}{4} + 1} \right)\ast 3 + 1 \\
 = x + 3 \\
 \end{array}
\]
\end{proof}

\begin{theorem}[The 4 successions theorem]
The action sequence 'TDDFDDFFBBBBTT' is
the succession of succession of succession of succession function for any
positive integer, i.e., equation (\ref{eq.m2.x+4}) always holds for any given positive integer.
\begin{equation}
\label{eq.m2.x+4}
'TDDFDDFFBBBBTT'(x)=TTBBBBFFDDFDDT(x) = x + 4
\end{equation}
\end{theorem}

\begin{proof}
The calculations are listed as follows,
\[
\begin{array}{l}
 'TDDFDDFFBBBBTT'(x) = TTBBBBFFDDFDDT(x)\\
 = \left( {\frac{\frac{\frac{(3x + 1)\ast 2\ast 2 - 1}{3}\ast 2\ast 2 -
1}{3} - 1}{3}\ast \frac{1}{2}\ast \frac{1}{2}\ast \frac{1}{2}\ast
\frac{1}{2}\ast 3 + 1} \right)\ast 3 + 1 \\
 = \left( {\frac{\frac{(4x + 1)\ast 4 - 1}{3} - 1}{16} + 1} \right)\ast 3 +
1 \\
 = x + 4 \\
 \end{array}
\]
\end{proof}

Because all the edges are two-way, there exists the precursor theorems corresponding to the succession theorems.

Of course, there exists more than one action sequence to perform the arbitrary successions and precursors.

Besides, these theorems above will be used in the model $M1$ to demonstrate the structures of $M1$ and $MS$.

\section{THE MODEL $M1$}

Compared with $M2$, model $M1$ only eliminates all non-integers from $M2$. Here, we also need to prove that
all positive integers are reachable in $M1$, i.e., $\Omega_1 =  \mathbb{N}^{+}$.

From now, we use a new notational method to represent a positive integer to facilitate the calculation. Basically, we use the 3-based numeral system.

For any given positive integer $A$ represented in the 3-based numeral system, if $A$ accepts a $'T'$ action, the value would be $A1$. Therefore, we use $A1$ to represent the number $3A + 1$. Formally, we use $A11$ to represent a number like $(9(A)_3 + 4)_{10}$.

To represent the carry in the 3-based numeral system, we use $(A+1)11$ to represent a number like $((A)_3 * 9 + 9 + 4)_{10}$.

We also use $A^D$ to represent the double value of $A$, i.e., the value after action ``D'', and $A^{D}11$ to represent a number like $(2*(A)_3*9 + 4)_{10}$.

We also use $A_1$ to represent $\lfloor A/2 \rfloor$, $A_2$ to $\lfloor \lfloor A/2 \rfloor /2 \rfloor$, and $A_3$ to $\lfloor \lfloor \lfloor A/2 \rfloor /2 \rfloor/2 \rfloor$. Obviously, $A_1$ is the value of $A$ after a $'B'$ action with consideration of the carry.

As to the $'F'$ action, we only need to erase the last $'1'$ symbol of $'A1'$.

Moreover, for positive integers $A$ and $C$, if there is an action sequence that can transform $A$ to $C$, we denote it as $A \Rightarrow C$.

\begin{definition} [9-cluster]For any given positive integer $A$, the set \{$A00$, $A01$, $A02$, $A10$, $A11$, $A12$, $A20$, $A21$, $A22$\}, i.e., in decimal, \{$9k+0$, $9k+1$, $9k+2$, $9k+3$, $9k+4$, $9k+5$, $9k+6$, $9k+7$, $9k+8$\} is called a 9-cluster. Here, $(A)_3 = (k)_{10}$.
\end{definition}

\begin{definition} [5-cluster]For any given positive integer $A$, the set \{$A00$, $A01$, $A02$, $A10$, $A11$\}, i.e., in decimal, \{$9k+0$, $9k+1$, $9k+2$, $9k+3$, $9k+4$\} is called a 5-cluster. Here, $(A)_3 = (k)_{10}$.
\end{definition}

\begin{definition} [3-cluster]For any given positive integer $A$, the set \{$A12$, $A20$, $A21$\}, i.e.,\{$9k+5$, $9k+6$, $9k+7$\} is called a 3-cluster. Here, $(A)_3 = (k)_{10}$.
\end{definition}

The proof procedure can be organized as follows:

\begin{enumerate}
  \item first we prove that every 5-cluster can form an internally connected subgraph from Lemma \ref{lemma.10.11} to Theorem \ref{th.5cluster}.
  \item next we prove that every 3-cluster can form an internally connected subgraph from Lemma \ref{lemma.20.21} to \ref{lemma.11.21.odd1} and Theorem \ref{th.3cluster.conn} .
  \item then we prove that every 3-cluster can connect to its corresponding 5-cluster.
  \item also we prove that every 9-cluster can form an internally connected subgraph.
  \item at last we show that all 9-clusters can connect to 1.
\end{enumerate}

Actually, there exist simpler proofs on $\Omega_1 =  \mathbb{N}^{+}$. However, to illustrate the structure of the Collatz graphs, we used the proof method stated above.

Here, we firstly prove that every 5-cluster can form an internally connected subgraph.

\begin{lemma}\label{lemma.10.11}The action sequence 'TDDFFBBT' can transform
$A10$ to $A11$.
\end{lemma}

\begin{proof}[The proof of $A10 \Rightarrow A11$.]
The calculations are listed as follows,

$T(A10)=A101$

$D(A101)= A^{D}202$

$D(A^{D}202)=(A^{DD}+1)111$

$F((A^{DD}+1)111) =(A^{DD}+1)11$

$F((A^{DD}+1)11) =(A^{DD}+1)1$

$B((A^{DD}+1)1)=A^{D}2$

$B(A^{D}2) = A1$

$T(A1)=A11$

\end{proof}

\begin{lemma}\label{lemma.11.10}The action sequence ('TDDFFBBT')$^{ - 1}$= 'FDDTTBBF'
can transform $A11$ to $A10$.
\end{lemma}

Because the procedure is inverse to that in Lemma \ref{lemma.10.11} and all edges in $M1$ is two-way, we omit the detailed proof procedure. Moreover, we will omit the proof procedure on all inverse lemmas.

\begin{lemma}\label{lemma.02.11}The action sequence 'DFFBTT' can transform
$A02$ to $A11$.
\end{lemma}

\begin{proof}[The proof of $A02 \Rightarrow A11$.]
The calculations are listed as follows,

$D(A02)=A^{D}11$

$F(A^{D}11)= A^{D}1$

$F(A^{D}1)= A^{D}$

$B(A^{D})=A$

$T(A)=A1$

$T(A1)=A11$

\end{proof}

\begin{lemma}\label{lemma.11.02}
The action sequence ('DFFBTT')$^{ - 1}$= 'FFDTTB' can
transform $A11$ to $A02$.
\end{lemma}

\begin{lemma}\label{lemma.01.11}The action sequence 'DDFFBBTT' can transform
$A01$ to $A11$.
\end{lemma}

\begin{proof}[The proof of $A01 \Rightarrow A11$.]
The calculations are listed as follows,

$D(A01)=A^{D}02$

$D(A^{D}02)= (A^{DD})11$

$F(A^{DD}11)= A^{DD}1$

$F(A^{DD}1)= A^{DD}$

$B(A^{DD})= A^{D}$

$B(A^{D})=A$

$T(A)=A1$

$T(A1)=A11$

\end{proof}

\begin{lemma}\label{lemma.11.01}
The action sequence ('DDFFBBTT')$^{ - 1}$= 'FFDDTTBB'
can transform $A11$ to $A01$.
\end{lemma}

\begin{lemma}\label{lemma.00.11}
The action sequence 'TDDFDDFFBBBBTT' can
transform $A00$ to $A11$.
\end{lemma}

\begin{proof}[The proof of $A00 \Rightarrow A11$.]
The calculations are listed as follows,

$T(A00)=A001$

$D(A001)= A^{D}002$

$D(A^{D}002)= A^{DD}011$

$F(A^{DD}011)= A^{DD}01$

$D(A^{DD}01)= A^{DDD}02$

$D(A^{DDD}02)= A^{DDDD}11$

$F(A^{DDDD}11)= A^{DDDD}1$

$F(A^{DDDD}1)= A^{DDDD}$

$B(A^{DDD})= A^{DD}$

$B(A^{DD})= A^{D}$

$B(A^{D})= A$

$T(A)=A1$

$T(A1)=A11$

\end{proof}

\begin{lemma}\label{lemma.11.00}
The action sequence ('TDDFDDFFBBBBTT')$^{-1}$ =
'FFDDDDTTBBTBBF' can transform $A11$ to $A00$.
\end{lemma}

\begin{theorem}[The 5-cluster connection theorem]\label{th.5cluster}For any given positive integer $A$, nodes $A00$, $A01$, $A02$, $A10$ and $A11$ are internally connected.
\end{theorem}

This theorem follows from Lemmas \ref{lemma.10.11}, \ref{lemma.11.10}, \ref{lemma.02.11}, \ref{lemma.11.02}, \ref{lemma.01.11}, \ref{lemma.11.01}, \ref{lemma.00.11} and \ref{lemma.11.00}.

\begin{theorem}[The 5-cluster attaching theorem] For any given positive integer $A$, there is at least a path from $A$ to a 5-cluster. \end{theorem}

\begin{proof}
$T(A) = A1$

$T(A1)=A11$
\end{proof}

From the 5-cluster attaching theorem, the action sequence $`TT'$ can assure that arbitrary positive integer
is connected to at least one 5-cluster.

Here, we prove that the 3-clusters are internally connected.

\begin{lemma}\label{lemma.20.21}
The action sequence 'TDDFFBBT' can transform
$A20$ to $A21$.
\end{lemma}

\begin{proof}[The proof of $A20 \Rightarrow A21$.]
The calculations are listed as follows,

$T(A20)=A201$

$D(A201)=(A^{D}+1)102$

$D((A^{D}+1)102)=((A^{D}+1)^{D})211$

$F(((A^{D}+1)^{D})211)=(A^{D}+1)^{D}21$

$F((A^{D}+1)21)=(A^{D}+1)^{D}2$

$B((A^{D}+1)^{D}2)=(A^{D}+1)1$

$B((A^{D}+1)1)=A2$

$T(A2)=A21$

\end{proof}

\begin{lemma}\label{lemma.21.20}The action sequence ('TDDFFBBT')$^{ - 1}$='FDDTTBBF'
can transform $A21$ to $A20$.
\end{lemma}

\begin{lemma}\label{lemma.12.21}The action sequence 'DDDFFBBTBT' can
transform $A12$ to $A21$.
\end{lemma}

\begin{proof}[The proof of $A12 \Rightarrow A21$.]

The calculations are listed as follows,

$D(A12)=(A^{D}+1)01$

$D((A^{D}+1)01)=(A^{D}+1)^{D}02$

$D((A^{D}+1)^{D}02)=(A^{D}+1)^{DD}11$

$F((A^{D}+1)^{DD}11)=(A^{D}+1)^{DD}1$

$F((A^{D}+1)^{DD}1)=(A^{D}+1)^{DD}$

$B((A^{D}+1)^{DD})=(A^{D}+1)^D$

$B((A^{D}+1)^{D})=(A^{D}+1)$

$T((A^{D}+1))=(A^{D}+1)1$

$B((A^{D}+1)1)=A2$

$T(A2)=A21$

\end{proof}

\begin{lemma}\label{lemma.21.12}The action sequence ('DDDFFBBTBT')$^{ - 1}$ =
'FDFDDTTBBB' can transform $A21$ to $A12$.
\end{lemma}

\begin{theorem}[The 3-cluster connection theorem]\label{th.3cluster}For any given positive integer $A$, nodes $A12$, $A20$ and $A21$ are internally connected.
\end{theorem}

According to Lemma \ref{lemma.20.21}, \ref{lemma.21.20}, \ref{lemma.12.21} and \ref{lemma.21.12}.

From now on, we will prove that every 3-cluster can connect to its corresponding 5-cluster.

\begin{lemma} When $A$ is even, there exists at least one
action sequence to transform $A21$ to $A11$, i.e., $A21 \Rightarrow A11$.
\end{lemma}

\begin{proof}[The proof of $A21 \Rightarrow A11$.]

The calculations are listed as follows,

$T(A21) = A211$

$B(A211) =A_{1}102$

$'DFFBTT'(A_{1}102)=A_{1}111$      (Lemma \ref{lemma.02.11})

$FFF(A_{1}111)=A_{1}$

$'DTT'(A_{1})=A11$

\end{proof}

\begin{lemma}\label{lemma.even.11.21}
When $A$ is even, there exists at least one sequence
to transform $A11$ to $A21$, i.e., $A11 \Rightarrow A21$.
\end{lemma}

\begin{lemma}\label{lemma.RD0.21.11}
When $A$ is odd, and $A=R0$, i.e., the last symbol of $A$ is $'0'$, there
exists at least one action sequence to transform $A21$ to $A11$, i.e., $R021$
$\Rightarrow$  $R011$.
\end{lemma}

\begin{proof}[The proof of $A21 \Rightarrow A11$.]
The calculations are listed as follows,

$R021\Rightarrow R^{D}112 $   (D)

$\Rightarrow  R^{D}11211$         (TT)

$\Rightarrow  R02102$       (B)

$\Rightarrow  R02111$       (Lemma \ref{lemma.02.11})

$\Rightarrow  R02$      (FFF)

$\Rightarrow  R11$      (Lemma \ref{lemma.02.11})

$\Rightarrow  R01$      (Lemma \ref{lemma.11.01})

$\Rightarrow  R011$         (T)

\end{proof}

\begin{lemma}\label{lemma.11.21}
When $A$ is odd, and $A=R0$, there exists at
least one action sequence to transform $A11$ to $A21$, i.e., $R011 $$\Rightarrow$
$R021$.
\end{lemma}

\begin{lemma}\label{lemma.21.11.odd1}
When $A$ is odd, and $A=R1$, there exists at
least one action sequence to transform $A21$ to $A11$, i.e., $R121 $$\Rightarrow$ $ R111$.
\end{lemma}

\begin{proof}[The proof of $A21 \Rightarrow A11$.]
The calculations are listed as follows,

$R121 \Rightarrow R112$      (Lemma \ref{lemma.21.12})

$\Rightarrow R11211$        (TT)

$\Rightarrow R_{1}02102$      \mbox{(B, here $R$ is even because $A$ is odd and $A=R1$)}

$\Rightarrow R_{1}02111$      (FFF)

$\Rightarrow R_{1}02$     (D)

$\Rightarrow R11$       (T)

$\Rightarrow R111$      (T)
\end{proof}

\begin{lemma}\label{lemma.11.21.odd1}
When $A$ is odd, and $A=R1$, there exists at least
one action sequence to transform $A11$ to $A21$, i.e., $R111 $$\Rightarrow$ $ R121$.
\end{lemma}

Here, we firstly discuss the relationship between $A22$ and $A11$ and then come
back to discuss the circumstance when $A$ is an odd and $A=R2$.

\begin{lemma}\label{lemma.22.11.even}
When $A$ is even, there exists at least one
action sequence to transform $A22$ to $A11$, i.e., $A22$$\Rightarrow$ $A11$.
\end{lemma}

\begin{proof}[The proof of $A22 \Rightarrow A11$.]

The calculations are listed as follows,

$B(A22) = A_{1}11$

$'FFDTT'(A_{1}11)=A11$

\end{proof}

\begin{lemma}\label{lemma.11.22.even}
When $A$ is even, there exists at least one
sequence to transform $A11$ to $A22$, i.e., $A11$$\Rightarrow$ $A22$.
\end{lemma}

\begin{lemma}\label{lemma.22.11.odd0}
When $A$ is odd, and $A=R0$, there
exists at least one action sequence to transform $A22$ to $A11$, i.e., $R022$
$\Rightarrow$  $R011$.
\end{lemma}

\begin{proof}[The proof of $R022 \Rightarrow R011$. ]
The calculations are listed as follows,

$R022 \Rightarrow R^{D}121$    (D)

$\Rightarrow R^{D}112$    (Lemma \ref{lemma.21.12})

$\Rightarrow R021$      (B)

$\Rightarrow R02$       (F)

$\Rightarrow R11$        (Lemma \ref{lemma.02.11})

$\Rightarrow R01$        (Lemma \ref{lemma.11.01})

$\Rightarrow R011$    (T)

\end{proof}

\begin{lemma}\label{lemma.11.22.odd0}
When $A$ is odd, and $A=R0$, there exists at
least one action sequence to transform $A11$ to $A22$, i.e., $R011 $$\Rightarrow$
$R022$.
\end{lemma}

\begin{lemma}\label{lemma.122.111}
When $A$ is odd, and $A=R1$, there exists at
least one action sequence to transform $A22$ to $A11$, i.e., $R122 $$\Rightarrow$ $ R111$.
\end{lemma}

\begin{proof}[The proof of $R122 \Rightarrow R111$.]
The calculations are listed as follows,

$R122 \Rightarrow R_{1}011$

$\Rightarrow R_{1}01$

$\Rightarrow R02$

$\Rightarrow R11$

$\Rightarrow R111$

\end{proof}

\begin{lemma}\label{lemma.111.122}
 When $A$ is odd, and $A=R1$, there exists at least
one action sequence to transform $A11$ to $A22$, i.e., $R111 \Rightarrow R122$.
\end{lemma}

\begin{lemma}\label{lemma.R02.R022}
There exists at least one action sequence to
transform $ R0\overbrace {2 \cdots 2}^n$ to $R0\overbrace {2 \cdots 2}^{n +
1}$, i.e., $R0\overbrace {2 \cdots 2}^n \Rightarrow R0\overbrace {2 \cdots 2}^{n +
1}$.
\end{lemma}

\begin{proof}
The calculations are listed as follows,

$R0$$\overbrace {2 \cdots 2}^n$$\Rightarrow$  $R0$$\overbrace {2 \cdots 2}^n$$1$

$\Rightarrow$ $R$$^{D}$1$\overbrace {2 \cdots 2}^{n - 1}$$12$

$\Rightarrow$  $R$$^{D}$1$\overbrace {2 \cdots 2}^{n - 1}$$21$

$\Rightarrow$ $R0$$\overbrace {2 \cdots 2}^{n - 1}$22$\Rightarrow$  $R0$$\overbrace {2 \cdots
2}^{n + 1}$
\end{proof}

\begin{lemma}\label{lemma.R022.R02}
There exists at least one action sequence to
transform $R0$$\overbrace {2 \cdots 2}^{n + 1}$ to $R0$$\overbrace {2 \cdots
2}^n$, i.e., $R0$$\overbrace {2 \cdots 2}^{n + 1}$$\Rightarrow$  $R0$$\overbrace {2 \cdots
2}^n$.
\end{lemma}

\begin{lemma}\label{lemma.R12.R122}
There exists at least one action sequence to
transform $R1$$\overbrace {2 \cdots 2}^n$ to $R1$$\overbrace {2 \cdots 2}^{n +
1}$, i.e., $R1$$\overbrace {2 \cdots 2}^n$ $\Rightarrow$  $R1$$\overbrace {2 \cdots 2}^{n +
1}$.
\end{lemma}

\begin{proof}
The calculations are listed as follows,

When $R$ is even,

$R1$$\overbrace {2 \cdots 2}^n$$\Rightarrow$  $R1$$\overbrace {2 \cdots 2}^n$$1$

$\Rightarrow$ $R$$_{1}$0$\overbrace {2 \cdots 2}^n$$2$

$\Rightarrow$  $R$$_{1}$0$\overbrace {2 \cdots 2}^n$$21$

$\Rightarrow$  $R$$_{1}$0$\overbrace {2 \cdots 2}^n$$12$

$\Rightarrow$ $R1$$\overbrace {2 \cdots 2}^n$$21$

$\Rightarrow$ $R1$$\overbrace {2 \cdots 2}^{n + 1}$

When $R$ is odd, let $R = (P+1)$,

$(P+1)1$$\overbrace {2 \cdots 2}^n$

$\Rightarrow$ $P$$_{1}$2$\overbrace {1 \cdots 1}^n$

$\Rightarrow$  $P$$_{1}$2$\overbrace {1 \cdots 1}^{n + 1}$

$\Rightarrow$ $(P+1)1$$\overbrace {2 \cdots 2}^{n + 1}$

Therefore, for any given $R$, $R1$$\overbrace {2 \cdots 2}^n$ $\Rightarrow$  R1$\overbrace {2
\cdots 2}^{n + 1}$

\end{proof}

\begin{lemma}\label{lemma.R122.R12}
There exists at least one action sequence to
transform $R1$$\overbrace {2 \cdots 2}^{n + 1}$ to $R1$$\overbrace {2 \cdots
2}^n$, i.e., $R1$$\overbrace {2 \cdots 2}^{n + 1}$$\Rightarrow$  $R1$$\overbrace {2 \cdots
2}^n$.
\end{lemma}

\begin{theorem}[The 2 appending theorem]\label{theorem.2.appending}
There exists at least one action sequence to
transform $R$$\overbrace {2 \cdots 2}^n$ to $R$$\overbrace {2 \cdots 2}^{n +
1}$, i.e., $R$$\overbrace {2 \cdots 2}^n$ $\Rightarrow$  R$\overbrace {2 \cdots 2}^{n +
1}$.
\end{theorem}

According to Lemma \ref{lemma.R02.R022} to \ref{lemma.R12.R122}, this theorem is obvious.

\begin{theorem}[The 2 backspace theorem]\label{theorem.2.backspace}
There exists at least one action sequence to
transform $R$$\overbrace {2 \cdots 2}^{n+1}$ to $R$$\overbrace {2 \cdots 2}^{n}$,
i.e., $R$$\overbrace {2 \cdots 2}^{n+1}$ $\Rightarrow$  $R$$\overbrace {2 \cdots 2}^{n}$.
\end{theorem}

According to Theorem \ref{theorem.2.appending}, this theorem is obvious.

\begin{lemma}\label{lemma.222.211}
When $A$ is odd, and $A=R2$, there exists at
least one action sequence to transform $A22$ to $A11$, i.e., $R222$ $\Rightarrow$  $R211$.
\end{lemma}

\begin{proof}[$R222$ $\Rightarrow$  $R211$.]

$'0R'$ should include at least one $'1'$ or $'0'$ before a series of $'2'$.

Therefore, $R222 \Rightarrow R2 \Rightarrow R211$, or say, $A22 \Rightarrow A \Rightarrow A11$.

\end{proof}

From Lemma \ref{lemma.22.11.even} to \ref{lemma.222.211}, we can obtain a conclusion as Theorem \ref{th.1cluster.conn}.

\begin{theorem}\label{th.1cluster.conn}
For any given positive integer $A$, there exists at least an action sequence to transform $A22$ to $A11$, i.e., $A22 \Rightarrow A11$.
\end{theorem}

This proposition follows from Lemma \ref{lemma.22.11.even} to \ref{lemma.222.211}.

\begin{lemma}\label{lemma.1cluster.conn.2}
For any given positive integer $A$, there exists at least an action sequence to transform $A11$ to $A22$.
\end{lemma}

Now, we can come back to discuss $A21$ when $A$ is odd.

\begin{lemma}\label{lemma.221.211}
When $A$ is odd, and $A=R2$, there exists at
least one action sequence to transform $A21$ to $A11$, i.e., $R221 \Rightarrow R211$.
\end{lemma}

\begin{proof}[Proof of $R221 \Rightarrow R211$.]
The calculations are listed as follows,

$R221 \Rightarrow R22$

$\Rightarrow R2$

$\Rightarrow R211$

\end{proof}

\begin{theorem}\label{th.3cluster.conn}
For any given positive integer $A$, there exists at least one action sequence to transform $A21$ to $A11$.
\end{theorem}

\begin{lemma}\label{lemma.3cluster.conn.2}
For any given positive integer $A$, there exists at least one action sequence to transform $A11$ to $A21$.
\end{lemma}

According to Theorems \ref{th.1cluster.conn}, \ref{th.3cluster.conn} and \ref{th.5cluster}, we can obtain Theorem \ref{th.9cluster}.

\begin{theorem}\label{th.9cluster}
For any given positive integer $A$, the corresponding 9-cluster is internally connected.
\end{theorem}

So we can prove that every positive integer can reach $(11)_3$, i.e., $(4)_{10}$; Of course, can also connect to $1$.

\begin{theorem}\label{theorem.A.11}
For any given positive integer $A$, there exists at least
one action sequence to transform $A$ to 11.
\end{theorem}

\begin{proof}
If the numbers of symbols $A$ is odd, then let $A = 0A$, i.e., add an
additional $0$ to the head of $A$, to make the numbers of symbols is even.

$A** \Rightarrow A11 \Rightarrow A$, here * is arbitrary one of $'0'$, $'1'$ and $'2'$.

By repeating this process, we obtain $A \Rightarrow 11$.

\end{proof}

According to the proofs above, we can obtain more conclusions.

\begin{lemma}\label{lemma.m1.descend}
For any given positive integer $A$, there exists at least
one action sequence $H$ in $M1$ such that $H(A) < A$.
\end{lemma}

\begin{lemma}\label{lemma.A.1}
For any given positive integer $A$, there exists at least one action
sequence to transform $A$ to $1$, i.e., $\Omega_1 = \mathbb{N}^{+}$.
\end{lemma}

Because $M1$ can be mapped into Fig. \ref{fig.m1}, Theorem \ref{theorem.A.11} and Lemma \ref{lemma.A.1} indicate that there is at
least one path from any given positive integer to $4$ or $1$.

\begin{theorem}[The node loop existence theorem]
For any given positive integer $A$, $A$ is in a loop.
\end{theorem}

\begin{proof}
1) If $A$ is odd, $A \Rightarrow A111$ $\Rightarrow A_{1}202$ $\Rightarrow A_{1}211$ $\Rightarrow A_{1}2$ $\Rightarrow  A_1$. Because $A_1 < A$ and $A_1 \Rightarrow 1$, therefore, $A$ is in a loop.

2) If $A$ is even, $A \Rightarrow A11$ $\Rightarrow A_{1}02$ $\Rightarrow A_{1}11$ $\Rightarrow A_{1}$. Because $A_1 < A$ and $A_1 \Rightarrow 1$, therefore, $A$ is in a loop.

\end{proof}

\section{THE MODEL $MS$}

Compared with model $M1$, all nodes in $M1$ are still in $MS$. However, some edges of model $MS$ become directed. That is, the Collatz graph in Fig. \ref{fig.ms} is a weakly connected graph and all the positive integers are weakly connected in it.

If we can prove that, for any given positive integer $A$, there exists an action sequence $H$ such that $H(A) < A$, then there should exist an path from $A$ to $1$, i.e., $\Omega_s = \mathbb{N}^{+} $, because $1$ is the smallest value.

\begin{lemma}\label{lemma.HA2.A2}
For any given $A$, there exists an action sequence $H$, such that $H(A2) < A2$.
\end{lemma}

\begin{proof}
If $A$ is even, then $B(A2) = A_{1}1$, $F(A_{1}1) = A_1$. Because $A_1 < A2$, this proposition holds.

If $A$ is odd, then $T(A2)=A21$, $B(A21)= A_{1}22$.
For any given $k>=1$, if $A_k$ is even, then $B\overbrace {BT \cdots BT}^{k}(A2)=A_{k+1}\overbrace{1 \cdots 1}^{k+1}$. Further, $\overbrace{F \cdots F}^{k+1}(A_{k+1}\overbrace{1 \cdots 1}^{k+1})= A_{k+1}$. Because $A_{k+1} < A2$, this proposition holds.

Because there exists a $k$ such that $A_k=0$, this proposition holds.
\end{proof}


\begin{lemma}\label{lemma.HA1.A1}
For any given $A$, there exists an action sequence $H$, such that $H(A1) < A1$.
\end{lemma}
\begin{proof}
$F(A1)= A < A1$
\end{proof}

\begin{lemma}\label{lemma.HA0.A0}
For any given $A$, there exists an action sequence $H$, such that $H(A0) < A0$.
\end{lemma}

\begin{proof}
If $A$ is even, $B(A0)$=$A_10<A0$, this proposition holds.
If $A$ is odd, $T(A0)=A01$, $B(A01)=A_{1}12$.
For $A_{1}12$, if $A_1$ is odd, $BF(A_{1}12) = A_{2}2 < A0$, this proposition holds; if $A_1$ is even, $TB(A_{1}12) = A_{2}022$. According to Lemma \ref{lemma.HA2.A2}, this proposition also holds.

Therefore, this proposition always holds.
\end{proof}

\begin{theorem}[The descending theorem]\label{theorem.descending}
For any given positive integer $A$, there exists an action sequence $H$ in $MS$ such that $H(A)<A$.
\end{theorem}
\begin{proof}
According to Lemmas \ref{lemma.HA0.A0}, \ref{lemma.HA1.A1} and \ref{lemma.HA2.A2}, this theorem holds.
\end{proof}

\begin{theorem}[The edge loop existence theorem]\label{theorem.edge.loop}
For any given positive even integer $A$, the edge $A \leftarrow A1$ is in a loop.
\end{theorem}

\begin{proof}
Because $A$ is even, $A1 \Rightarrow A11 \Rightarrow A_{1}02$.
\begin{enumerate}
\item If $A_1$ is even, $A_{1}02 \Rightarrow A_{2}01$ $\Rightarrow A_{2}0$ .
Because $A_20 < A$ and $A_1 \Rightarrow 1$, the edge $A \leftarrow A1$ is in a loop.

\item If $A_1$ is odd, $A_{1}02 \Rightarrow A_{1}021$ $\Rightarrow A_{2}122$.

\item If $A_2$ is odd, $A_{2}122 \Rightarrow A_{3}211$ $\Rightarrow A_{3}2$. Because $A_{3}2 < A$ , this theorem holds.

\item If $A_2$ is even, $A_{2}122 \Rightarrow A_{3}222$. By repeating this analysis process,
we obtain $A_{3}222 \Rightarrow A_{k+1}$ because the result should finally is even in some time $k$. Since $A_{k+1} < A$, this theorem holds.

\end{enumerate}

Therefore, the edge $A \leftarrow A1$ is in a loop.
\end{proof}

%

\begin{lemma}\label{lemma.m2.A.1}
For any given positive integer $A$, there exists a path from $A$ to 1.
\end{lemma}

\begin{lemma}\label{lemma.omaga.N}
In model $MS$, $\Omega_s=\mathbb{N}^{+}$.
\end{lemma}

\begin{definition}[The removable edge]
In Fig. \ref{fig.ms} (model $MS$), when an edge $e$ is removed, $\Omega_s$ does not change, then the edge $e$ is said removable.
\end{definition}

\begin{theorem}[The de-looping theorem]
For any given edge $e$, if $e$ belongs to the edge set $E_1$ indicated by $\frac{x-1}{3}$ if $x \bmod 6 \equiv 1$, then $e$ is removable.
\end{theorem}

\begin{proof}
According to Theorem \ref{theorem.edge.loop}, and notice that the edge $e$ has a different direction to the other path is the loop, therefore, $e$ is removable.
\end{proof}

\section{THE MODEL $M0$}

The graph of $M0$ is a directed graph. Compared with the model $MS$, model $M0$ eliminates all the edges indicated by the action $F$.

\begin{lemma} [The Collatz conjecture]For any given positive integer $A$, $A$ is reachable in $M0$.
\end{lemma}

\begin{proof}
According to Lemma \ref{lemma.omaga.N}, all positive integers which exist in Fig. \ref{fig.ms} (model $MS$) still exist in Fig. \ref{fig.m0} (model $M0$).

According to the de-looping theorem in model $MS$, all these edges in $E_1$ can be removed one by one starting from $7 \rightarrow 2$.

After the removals of the edge set $E_1$, all the edges belonging to the edge set $E_4$, which is indicated by $\frac{x-1}{3}$ if $x \bmod 6 \equiv 4$, would be abundant (have no successions nodes) in $MS$, and hance are removable.

After the removals of the edge sets $E_1$ and $E_4$, $MS$ becomes $M0$. Therefore, the Collatz conjecture holds.
\end{proof}

\begin{theorem}
There is only one circle $4 \rightarrow 2  \rightarrow 1 \rightarrow 4$ in $M0$.
\end{theorem}

\begin{proof}
Because every node in $M0$ has only an out-link, according to graph theory and Fig. \ref{fig.m0} (model $M0$) is a connected graph, only the positive integer $1$ has an
extra out-link, and it out-links to the positive integer $4$, so only
one circle $4 \rightarrow 2  \rightarrow 1 \rightarrow 4$ exists.
\end{proof}

\section{CONCLUSIONS}

This paper proves the $3x+1$ problem. The result shows that all the positive
integers can be transformed to $1$ by the iteration of $f$. Equivalently, there are no other cycles other than $4 \rightarrow 2  \rightarrow 1 \rightarrow 4$ cycle and there is no divergent trajectories.

The result in this paper would be useful to the research of chaos \cite{27,28}, computer
science \cite{10}, complex systems and so on.

\textbf{Acknowledgements.} The authors are grateful to Dr. CL Zhou, Prof. YX Li, Prof. J Qin, Dr. BB Wang, Dr. WW Wang, Mr. QF Wang, Mr. GC Tang and Mr. JW Zheng for discussions and inspirations. The authors thank to the supports from the State Key Laboratory of Networking and Switching Technology (No. SKLNST-2010-1-04) and the State Key Laboratory of Software Engineering (No. SKLSE2012-09-15) and the Fundamental Research Funds for the Central Universities (No. CZY13010). This paper is also supported by the China Scholarship Council.

\bibliographystyle{amsplain}
\bibliography{TP1}

%
%
%
%
%
%
%

\end{document}